\documentclass{IEEEtran}

\normalsize

\ifCLASSINFOpdf
\else
\fi

\usepackage[T1]{fontenc}
\usepackage{mathrsfs}
\usepackage{bm}
\usepackage{amsmath}
\usepackage{breqn}
\usepackage{amsthm}
\usepackage{amssymb}
\usepackage{graphicx}
\PassOptionsToPackage{normalem}{ulem}
 \usepackage{subfig}
\usepackage{setspace}
\usepackage{cite}
\usepackage{xcolor}
\usepackage[normalem]{ulem}
\useunder{\uline}{\ul}{}
\usepackage{algorithm}
\usepackage{algpseudocode}
\makeatletter

\theoremstyle{plain}
\newtheorem{theorem}{\protect\theoremname}
 \theoremstyle{plain}
 \newtheorem{conjecture}{\protect\conjecturename}
 \theoremstyle{plain}
 
 \theoremstyle{plain}
 \newtheorem{definition}{\protect\definitionname}
  \theoremstyle{definition}
  \newtheorem{lemma}{\protect\lemmaname}
 \theoremstyle{remark}
 \newtheorem{remark}{\protect\remarkname}
  \newtheorem{assumption}{\protect\assumptionname}
\theoremstyle{assumption}
  \theoremstyle{proposition}

\theoremstyle{algorithm}

\setkeys{Gin}{width=1.0\columnwidth}

\makeatother

 \providecommand{\definitionname}{Definition}
 \providecommand{\lemmaname}{Lemma}
 \providecommand{\propositionname}{Proposition}
 \providecommand{\remarkname}{Remark}
\providecommand{\theoremname}{Theorem}
\providecommand{\conjecturename}{Conjecture}
\providecommand{\assumptionname}{Assumption}

\begin{document}
\vspace{-1cm}
 \title{ Whittle's index-based age-of-information minimization in multi-energy harvesting source networks
 \thanks{The authors are with the Department of Electrical Engineering, Indian Institute of Technology, Delhi. Email:  \{akanksha.jaiswal, arpanc\}@ee.iitd.ac.in.   }
 \thanks{A. C. acknowledges support via grant no. GP/2021/ISSC/022 from I-Hub Foundation for Cobotics and grant no. CRG/2022/003707 from Science and Engineering Research Board (SERB), India.}
 \thanks{A preliminary shorter version of this paper is under review in a conference.}
 }
% \thanks{A. C. acknowledges support via the professional development fund and professional development allowance from IIT Delhi, grant no. GP/2021/ISSC/022 from I-Hub Foundation for Cobotics and grant no. CRG/2022/003707 from Science and Engineering Research Board (SERB), India.}}

%\thanks{The preliminary version of this paper appeared in \cite{jaiswal2021minimization}.}}

 \author{
Akanksha Jaiswal, Arpan Chattopadhyay \\
\vspace{-0.8cm}
 }
 
\maketitle

\begin{abstract} 
We consider the problem of source sampling and transmission scheduling for age-of-information minimization in a system consisting of multiple energy harvesting (EH) sources and a sink node. At each time, one of the sources is selected by the scheduler and the quality of its channel to the sink is measured. This probed channel quality is then used to decide whether a source will sample an observation and transmit the packet to the sink in that time slot. We formulate this problem as a constrained Markov decision process (CMDP)   assuming i.i.d. energy arrival and channel fading processes, and relax it using a  Lagrange multiplier. We apply a near optimal Whittle's index policy to decide the node to be probed. Next, for the probed node,  we derive an optimal threshold policy, which recommends source sampling and observation transmission from the probed source only when the measured channel quality is above a threshold. Our proposed policy is called Whittle's index and threshold based source scheduling and sampling  (WITS3) policy. However, in order to calculate Whittle's indices, one must be aware of the underlying processes' transition matrices, which are occasionally concealed from the scheduler. Therefore, we further propose a variant Q-WITS3 of WITS3 based on Q-learning assisted by two timescale asynchronous stochastic approximation,   which seeks to learn Whittle's indices and optimal policies for the case with unknown channel states and EH characteristics. Numerical 
results demonstrate the efficacy of our algorithms over two baseline policies.
\end{abstract}


\vspace{-8pt}
\section{Introduction}\label{section:introduction}
 Many real-time applications of wireless communication systems such as the Internet of Things (IoT) and cyber-physical systems (CPS) require fresh information at the monitoring devices. To this end, the popular Age of Information (AoI) metric  \cite{kaul2012real}   characterizes the freshness of information for status updates and is defined as the time elapsed since the reception of the previous update at the monitoring node. 
 However, two key challenges in such real-time status update systems are battery limits and limited communication bandwidth.  In order to deal with the scarcity of bandwidth, multiple sources often share a common frequency band for transmitting their time-sensitive status updates to a monitoring node.  On the other hand, energy harvesting source (EH) nodes can address the challenge of battery limits. However, smart scheduling policies need to be used to avoid interference among these nodes and to optimally exploit the EH feature available at the source nodes.\\ 
 AoI has gained significant popularity as a network performance metric, especially for ultra-reliable low-latency communication. For example, the papers \cite{bacinoglu2018achieving,ceran2021learning, abd2020aoi, jaiswal2021minimization, jaiswal2023age} have proposed optimal status update policies to minimize AoI for a single EH source under various network settings; especially see  \cite{bacinoglu2018achieving} for age-energy trade-off and optimal threshold policy for AoI minimization, and  \cite{ceran2021learning} for a threshold policy to minimize AoI for a single EH sensor with hybrid automatic repeat request (HARQ) protocol. The authors of  \cite{abd2020aoi} have proposed an age-optimal threshold policy for an RF-powered source node sending status updates to the destination considering the case where the source can either harvest energy or send data at a time. Our paper \cite{jaiswal2021minimization, jaiswal2024age}  formulated the AoI minimization problem as an MDP for a single EH source with channel probing capability and showed that the optimal channel probing and source sampling policies exhibit threshold structure.  Another paper \cite{jaiswal2023age} from us proposed reinforcement learning algorithms to minimize AoI for an EH source working in a non-stationary environment with unknown, time-varying channel statistics and energy arrival rate.\\ 
 On the other hand,   AoI minimization in networks consisting of multiple source nodes contending for a shared channel has been considered too  \cite{zakeri2023minimizing, wang2024scheduling, bedewy2019age,xie2021reinforcement, zhou2019joint, hatami2022demand,tang2020minimizing}. The paper \cite{zakeri2023minimizing} has considered a multi-source relaying system with 
independent sources, random packet generation, buffer-aided transmitter/full-duplex relay, and error-prone links, and minimized the weighted average AoI subject to transmission constraints via CMDP formulation. 
The authors of \cite{bedewy2019age} have considered   AoI minimization over an error-free delay channel. They found that the optimal scheduling and sampling policies can be decoupled. They proposed the Maximum Age First (MAF) scheduling policy and Zero wait sampling policy to minimize the total average peak age, whereas the optimal sampling policy for minimizing the total average age turned out to be a threshold policy. 
The paper \cite{zhou2019joint} has examined the optimal sampling and updating processes for IoT devices under sampling/updating cost,   a block-fading channel, and average energy constraint. For the single source case, they proposed a threshold-type sampling and updating policy. For the multiple source case, they provided an approximate policy via Q-functions. 
The authors in \cite{hatami2022demand} have proposed a relax-then-truncate optimal policy to minimize average on-demand AoI for a resource-constrained IoT network with a cache-enabled edge node between users and EH sensors. The authors of  \cite{tang2020minimizing} have designed a scheduling algorithm to minimize AoI for multi-sensor network with time-varying channel states under bandwidth and power consumption constraints. They decouple the problem across sensors and propose a threshold policy for scheduling of each sensor. The work \cite{ceran2021reinforcement} has addressed the AoI minimization scheduling problem for a system model where a source node sends status updates to multiple users via error-prone channels under average transmission constraints. For known channel statistics, it investigated optimal scheduling policy for both automatic repeat request (ARQ) and hybrid ARQ (HARQ) protocols.\\
%whereas, for the unknown case, it studied different RL techniques (UCRL2, DQN, average-cost SARSA with LFA) and compared them numerically.  }
However, in order to reduce the computational complexity in solving constrained MDPs ({\em i.e.,} CMDPs) for multi-source systems,   Whittle's Index \cite{whittle1988restless}   for restless multi-armed bandits (RMABs)  has been widely used since it provides a near-optimal solution. Many papers have used  Whittle's Index based policy for AoI minimization problems in various multi-source scheduling systems \cite{kadota2016minimizing, sun2019closed, hsu2018age, tang2021whittle, tong2022age, kriouile2021global, tripathi2019whittle}, but without considering energy harvesting sources and channel probing. 
 For example, the paper \cite{kadota2016minimizing} considered a single base station (BS)- multiple client broadcast network.  For the symmetric network case, the authors showed that greedily transmitting a packet with the highest age is optimal, while for the general network case, they established that the problem is indexable and obtained Whittle's index in closed form.
The paper \cite{sun2019closed} has considered   AoI minimization scheduling for multiple source-single sink setup,     derived Whittle’s indices in closed form, and proposed an Index-Prioritized Random Access (IPRA) scheme. The authors of  \cite{tang2021whittle} have developed Whittle's index-based scheduling policy to minimize a non-decreasing cost function of AoI for a multi-user system model where a BS stores stochastically generated packets in buffers and sends it via unreliable channels. 
The paper \cite{kriouile2021global} has studied the average age minimization problem for a system model where only a fraction of the users can transmit data packets simultaneously to a BS over unreliable channels and the BS can decode these packets with some success probabilities. It used a Cauchy criterion to prove the optimality of Whittle’s index policy. 
In \cite{tripathi2019whittle}, Whittle's index approach has been used to find scheduling
policies that minimize nonlinear cost functions of AoI for a single-hop multiple source-single BS network.\\
Information on transition dynamics, which are frequently unknown in advance, is necessary to solve RMABs. The papers \cite{ceran2021reinforcement, wang2023optimistic, mao2020model, akbarzadeh2023learning, xiong2022learning, fu2019towards, DBLP:journals/corr/abs-2004-14427, killian2021q} have proposed various online learning techniques for efficient planning in RMAB scenarios with unknown transitions. For example, the authors in \cite{wang2023optimistic} have proposed a Whittle's index policy-based algorithm named the UCWhittle algorithm which uses the upper confidence bound to learn transition probabilities. Different RL techniques such as UCRL2, DQN, and average-cost SARSA with LFA  have been studied and compared numerically in \cite{ceran2021reinforcement} for the unknown channel statistics. The paper \cite{mao2020model} has considered reinforcement learning in non-stationary systems where the reward and state distributions vary with time such that they do not surpass their respective variation budget. It proposed Restart Q-UCB (Q-learning with UCB) and Double Restart Q-UCB algorithms where the latter doesn't need to know the variation budget in advance.  The authors in \cite{akbarzadeh2023learning} have proposed Thompson-sampling based learning algorithm for a restless bandit system model.  Two index-aware RL algorithms named GM-R2MAB and UC-R2MAB have been presented for infinite-horizon average-reward RMABs in \cite{xiong2022learning}.
The papers \cite{fu2019towards, DBLP:journals/corr/abs-2004-14427, killian2021q} have utilized  Q-learning algorithms to compute the Whittle's indices for RMAB, and especially \cite{DBLP:journals/corr/abs-2004-14427} have proved convergence of their two-timescale stochastic approximation based algorithm under the assumption of indexability. Unlike the literature, our Q-WITS3 algorithm learns the Q-value functions and Whittle's indices for a two-stage action model involving source probing and sampling decisions. Though our previous work \cite{jaiswal2024age} had considered a single source with probing capabilities and proposed a Q-learning algorithm for a two-stage action model,    it did not involve Whittle's index-based decision.    \\
In contrast to the existing literature, our main contributions in this paper are as follows:

 \begin{itemize}    
\item We study the AoI minimization problem in multiple EH source system with i.i.d. energy arrival process and i.i.d. channel at each source.    
At a time,   two decisions need to be made optimally: (i) which source to select for channel probing, and  (ii) whether to sample and transmit the status update from the selected source based on the channel quality.   We formulate this problem as a CMDP where only one source can be scheduled at a time. However, a key challenge here is the two-stage action model-probing and sampling.  By using the Lagrangian relaxation method, we convert the CMDP   to an unconstrained MDP and further decouple this optimization problem across sources. For the unconstrained MDP subproblems, we propose a Whittle's index and threshold based source scheduling and sampling policy (WITS3 policy) where, at each time, the source node having the highest Whittle's index is probed, and sampled based on the channel probe outcome. We prove that, for each decoupled MDP sub-problem, 
the optimal sampling policy amounts to checking whether the probed channel quality is above a threshold.

\item We propose the Q-WITS3 algorithm that computes the optimal source scheduling and sampling policies for the unconstrained MDP subproblems of a multiple-source system with unknown channel statistics and energy harvesting characteristics. In contrast to the standard Q-learning algorithm which takes a single action at each time instant, we basically combine  WITS3 and Q-learning that seeks to learn the Q-values and Whittle's indices under the two-stage action model. It is important to note that the algorithm involves two timescale asynchronous stochastic approximation; the Q-values are updated in the faster timescale and the Whittle's indices are updated at a slower timescale.
\end{itemize} 

%\subsection{Organization}\label{section:Organization}
%The rest of the paper is organized as follows. The system model is presented in Section \ref{section:system-model}.
%The formulation of the AoI minimization problem and its proposed solution is given in Section \ref{section:Problem formulation and proposed policy}. In Section \ref{section:Simulation results}, simulation results are provided. Finally, conclusions are drawn at Section \ref{section:Conclusion}.

\vspace{-6pt}

\section{System model}\label{section:system-model} 
We consider a multiple source single sink sensing network that consists of $N$  EH sources measuring $N$ processes and sending the observation packets to a sink node over a fading channel, as in Fig.~\ref{multis-system-model}. Time is discretized such that  $t \in \{0, 1,2,3, \cdots\}$. At each time,  only one out of the $N$ source nodes can be probed to estimate the quality of its channel to the sink.  After probing, the selected source node can further take decision on whether to sample a process and transmit the data packet to the sink, based on the measured channel quality. We assume that probing requires a negligible amount of energy since the data packets that need to be transmitted are usually very large in size. However, if a source samples the process and communicates the packet to the sink, $E_s$ units of energy are used.

The source nodes harvest energy from the environment. Each source~$i \in \{ 1,2,3, \cdots N \}$ is equipped with a finite energy buffer of size $B_i$ units.   If the   $i$-th source node already has $B_i$ units of energy, the newly generated energy packets will be lost and not stored in the energy buffer of source $i$ unless that source spends some energy in sensing and communication.
At time~$t$, $A_{i}(t)$ denotes  the number of energy packet arrivals at source node $i$, and $E_{i}(t) \in \{0,1,\cdots, B_{i}\}$ denotes the available energy at source   $i$. We model   the energy packet generation process $\{A_{i}(t)\}_{t \geq 0}$ for the $i$-th source   as an i.i.d.  Bernoulli process with known mean $\lambda_i>0$; i.e., $\mathbb{E}(A_{i}(t))=\lambda_{i}$.  % \textcolor{blue}{see %\cite{hatami2021aoi, gindullina2021age, leng2019age, chen2019age} for reference}. 

\begin{figure}[t!]
  \begin{center}
 \includegraphics[height=3.8cm,width=7cm]{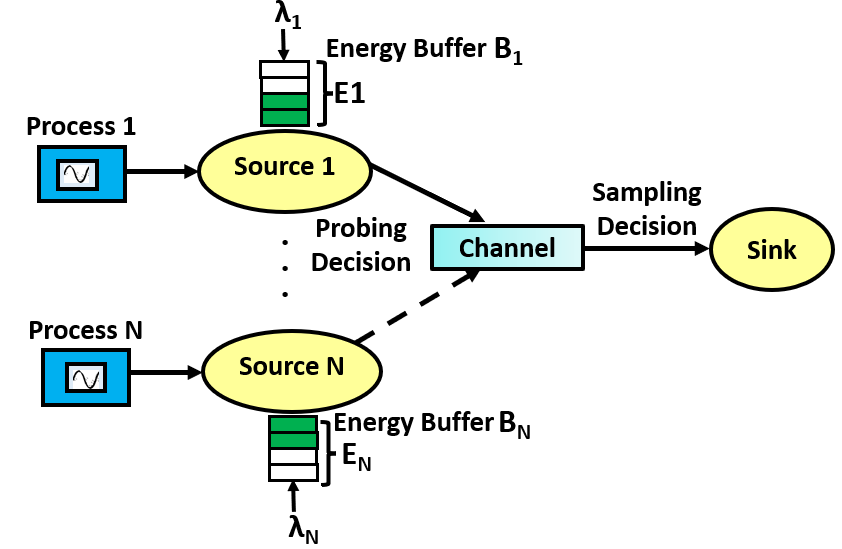}
 \caption{Remote sensing system with multiple EH sources.}
 \label{multis-system-model}
 \end{center}
 \vspace{-5mm}
\end{figure}

A fading channel is considered between the $i$-th source and the sink node. The term $C(i, t)$ is used to denote the channel state between the $i$-th source and the sink at time~$t$, and we assume that $C(i, t) \in \{C_1, C_2,\cdots, C_m\}$, where $m$ is the finite number of channel states. Also,  $p(i,t)$ denotes  the  successful packet transmission probability from the $i$-th source to the sink node at time~$t$, and it belongs to the set $p(i,t) \in \{p_1, p_2, \cdots, p_m\}$, $\forall i$. The term  $r_{i}(t) \in \{0,1\}$ is an indicator of successful packet transmission from the $i$-th source to the sink at time~$t$. Therefore,  $\mathbb{P}(r_{i}(t)=1|C(i, t)=C_j)=p_j$ and $\mathbb{P}(p(i,t)=p_j|C(i, t)=C_j)=1$ for all $j \in \{1,2,\cdots,m\}$ and $i \in \{1,2,\cdots, N\}$.  Here, we assume that $\{C(i, t)\}_{t \geq 0}$ is i.i.d. across $t$ for each $1 \leq i \leq N$, and that the channel states across different nodes are independent. Let $q_{j,i} \doteq \mathbb{P}(C(i, \cdot)=C_j), 1 \leq j \leq m$, denote the channel state distribution of source $i$ to the sink node. It is assumed that, at each time~$t$, only one out of $N$ source-sink channels is probed to find its state $C(i,t)$.

At each time~$t$, let us denote by $b_{i}(t)$   the indicator of selecting source $i$ to probe the channel, thus   $b_{i}(t)\in \{0,1\}$,  $\forall i$. In order to avoid a collision,  only one source is selected to probe the channel (and consequently decide whether to sample and transmit status update or not). This induces the constraint $\sum_{i=1}^N  b_{i}(t) \leq 1 ,   \forall t \geq 1$.

Also, for each source $i$, let $a_{i}(t) \in \{0,1\}$ denote the action of deciding whether to sample and communicate the data packet from the $i$-th source to the sink node (action $a_{i}(t)=1$), or to remain idle (action $a_{i}(t)=0$). Thus, at any time $t$, if $b_{i}(t)=0$, then  $a_{i}(t)=0$ as well;  this means that any source whose channel is not probed does not transmit. On the other hand, if $b_i(t)=1$, then $a_i(t) \in \{0,1\}$. Thus, the action space of source $i$ is given by  $\mathcal{A}_{i}=\{(0,0), (1,0), (1,1)\}$, where a generic action for the $i$-th source at time~$t$ is denoted by $(b_{i}(t),a_{i}(t))$. For $a_{i}(t)=1$,  $E_{s}$ energy units from source $i$ are used to sample and transmit status update. At any time $t$, if the available buffer energy for the $i$-th source is less than that required for sampling, i.e., if $E_{i}< E_s$, then the $i$-th source cannot be selected for channel probing because it does not have enough energy to sample a process and transmit.

Let us define $\kappa_{i}(t)\doteq \sup\{0 \leq \kappa < t: a_{i}(\kappa)=1, r_{i}(\kappa)=1\}$ be the last time instant before time~$t$, when the $i$-th source had successfully sent an observation packet to the sink. The AoI for the $i$-th source at time~$t$ is  given by $K_i(t)=(t-\kappa_i(t))$. Nevertheless, if $a_{i}(t)=1$ and $r_{i}(t)=1$, then $K_i(t)=0$ because the current status update of the $i$-th source is available to the sink node.

Our goal is to find a  stationary source scheduling and sampling policy $\bm{\pi}$ that minimizes the expected sum AoI across all the sources averaged over time, under  the source scheduling constraint for probing at each time:

\begin{align}\label{eqn:main-problem}
\footnotesize
 \min_{\bm{\pi}}  \lim_{T \rightarrow \infty}\frac{1}{TN} \sum_{t=1}^T \sum_{i=1}^N \mathbb{E}_{\bm{\pi}} (K_i(t))\nonumber\\
 s.t. \hspace{0.2 cm } \sum_{i=1}^N  b_{i}(t) \leq 1, \hspace{0.6 cm } \forall t \geq 1
 \normalsize
\end{align}

%\vspace{-2.5cm}

% Please add the following required packages to your document preamble:

\section{Problem formulation and proposed policy} \label{section:Problem formulation and proposed policy}
 Here, we derive the optimal policy for channel probing and source sampling and communication for multiple  EH sources. It is important to note that we assume the existence of a centralized scheduler that has access to the energy level and age of all nodes at each time instant, based on which it decides which source node will probe its channel quality to the sink.

 We formulate  problem~\eqref{eqn:main-problem} as a long-run average cost MDP. At time $t$, let us denote by $s_{i} \in  \mathcal{S}_{i}$ a   generic state of the  $i$-th source,  which is given by  $s_{i} =(E_{i}, K_{i})$. Here,     $E_{i} \in \{0, 1, . . .,B_{i}\}$ is the available buffer energy of source $i$, and  $K_{i}$ is the AoI of source node $i$. Thus $\mathcal{S}_{i} \doteq \{0,1,\cdots,B_i\} \times \mathbb{Z}_+$ is the   state space for source~$i$. At any time $t$, if the $i$-th source is selected to probe the channel (i.e., $b_i(t)=1$)), it encounters a generic intermediate state $v_{i} = (E_{i},K_{i},C(i, \cdot))$ with the corresponding state space $\mathcal{V}_{i}=\{0, 1,\cdots,B_i\}\times \mathbb{Z}_{+} \times \{C_{1}, C_{2},\cdots, C_{m}\}$; this generic intermediate state   additionally contains   the current channel state $C(i, \cdot)$ obtained by probing the $i$-th source node. The system state at any time $t$ is defined as $\bm{s}(t) = (s_{1}(t), . . . , s_{N}(t)) \in \mathcal{S}$ where $\mathcal{S} = \mathcal{S}_{1} \times  \cdots, \times \mathcal{S}_{N}$. The centralized scheduler knows $\bm{s}(t)$ at time~$t$, and uses this information to compute the probing decision $\{b_i(t)\}_{1 \leq i \leq N}$ such that $\sum_{i=1}^N b_i(t)=1$. Then a source $i$ with $b_i(t)=1$ chooses  $a_i(t) \in \{0,1\}$ and acts accordingly.

\subsection{MDP Formulation } \label{section:single-sensor-single-process_stationary}
 As mentioned earlier, we use the MDP theory to formulate the AoI minimization problem. At time~$t$,  if $b_{i}(t) = 0, a_{i}(t) = 0$,   the expected single-stage AoI cost for source node  $i$ is $c_i(s_{i}(t), b_{i}(t)=0, a_{i}(t)=0) = K_{i}$. However, at any time $t$, if the $i$-th source node is selected to probe the channel state, then it encounters an intermediate state $v_i(t)$ that contains the channel state $C(i,t)$ from source~$i$ to the sink. Then, the expected single-stage AoI cost for source node $i$, conditioned on $C(i,t)$, is given by $c_i(s_{i}(t), b_{i}(t)=1, a_{i}(t)=0) = K_{i}$, and $c_i(s_{i}(t), b_{i}(t)=1, a_{i}(t)=1) = K_{i}(1-p(C(i, t))$, where the expectation is taken over $r_i(t)$ conditioned on $C(i,t)$.

 The $i$-th source state transition probability is denoted by $\mathbb{P}(s_{i}(t+1)| s_{i}(t), b_{i}(t), a_{i}(t))$ which gives the probability of state transition from state $s_{i}=(E_{i}, K_{i})$ under   action $ (b_{i}, a_{i})$   to next state $ s^'_{i}=(E^'_{i}, K^'_{i})$  at time $t+1$ and is defined by:

\footnotesize
\begin{eqnarray}\label{eqn:state_transition_1}
&&\mathbb{P}(s^'_{i}|s_{i}=(E_{i}\gtreqless E_{s}, K_{i}), b_{i}=0, a_{i}=0) \nonumber\\
&=&
 \begin{cases}
 \lambda_{i}, \hspace{1 cm} s^'_{i}= (E^'_{i}=\min \{E_{i}+1,B_i\}, K^'_{i}= K_{i}+1 )  \\ \hspace{4.6cm}  \cdots \bm{\text{Case } 1} \\
 (1- \lambda_{i}),  \hspace{0.3cm} s^'_{i}=(E^'_{i}= E_{i},  K^'_{i}=K_{i}+1)  \\ \hspace{2.7cm}  \hspace{2.1cm}\cdots \bm{\text{Case } 2} \\
 0,  \hspace{2.1cm}\text{otherwise} \hspace{1.1cm}\cdots \bm{\text{Case } 3}
 \end{cases}
\end{eqnarray}
\normalsize

\footnotesize
\begin{eqnarray}\label{eqn:state_transition_2}
&&\mathbb{P}(s^'_{i}|s_{i}=(E_{i}\geq E_{s}, K_{i}), b_{i}=1, a_{i}=0) \nonumber\\
&=&
 \begin{cases}
  q_{j,i}\lambda_{i}, \hspace{0.8cm} s^'_{i}= (E^'_{i}=\min \{E_{i}+1,B_i\}, K^'_{i}= K_{i}+1 ) \\\hspace{5.1cm} \cdots  \bm{\text{Case } 4} \\
  q_{j,i} (1- \lambda_{i}),  \hspace{0.1cm} s^'_{i}=(E^'_{i}= E_{i},  K^'_{i}=K_{i}+1)  \\ \hspace{2.7cm}  \hspace{2.4cm}\cdots \bm{\text{Case } 5} \\
 0,  \hspace{2.2cm}\text{otherwise} \hspace{1.1cm}\cdots \bm{\text{Case } 6}
 \end{cases}
\end{eqnarray}
\normalsize

\footnotesize
\begin{eqnarray}\label{eqn:state_transition_3}
&&\mathbb{P}(s^'_{i}|s_{i}=(E_{i} \geq E_{s}, K_{i}), b_{i}=1, a_{i}=1) \nonumber\\
&=&
 \begin{cases}
 q_{j,i}\lambda_{i}p_{j}, \hspace{0.8cm} s^'_{i}= (E^'_{i}=\min \{E_{i}-E_{s}+1,B_i\},K^'_{i}= 1) \\ \hspace{3cm}  \hspace{2.8cm}\cdots \bm{\text{Case } 7} \\
 q_{j,i}(1-\lambda_{i})p_{j}, \hspace{0.1cm} s^'_{i}= (E^'_{i}=E_{i}-E_{s}, K^'_{i}= 1) \cdots \bm{\text{Case } 8} \\
 q_{j,i} \lambda_{i}(1-p_{j}),  \hspace{0.1cm} s^'_{i}=(E^'_{i}=\min \{E_{i}-E_{s}+1,B_i\}, \\ \hspace{2.8cm}K^'_{i}=K_{i}+1)
  \hspace{1.2cm} \cdots\bm{\text{Case } 9} \\
  q_{j,i} (1-\lambda_{i})(1-p_{j}),  \hspace{0.1cm} s^'_{i}=(E^'_{i}=E_{i}-E_{s}, K^'_{i}=K_{i}+1 )\\
   \hspace{5.5cm} \cdots\bm{\text{Case } 10}  \\
 0,  \hspace{2.9cm}\text{otherwise} \hspace{1.2cm}\cdots \bm{\text{Case } 11}
 \end{cases}
\end{eqnarray}
\normalsize

The above equations \eqref{eqn:state_transition_1}-\eqref{eqn:state_transition_3} follow from the independence of the energy harvesting process and channel statistics of source $i$ from other source nodes. Note that, in \eqref{eqn:state_transition_2}-\eqref{eqn:state_transition_3} when  $b_{i}=1$, the $i$-th source will encounter intermediate state $(E_{i}, K_{i}, C(i, \cdot))$ and the distribution of the next state will depend on the channel state occurrence probabilities $\{q_{j,i}\}_{1 \leq j \leq m}$, packet success probabilities $\{p_j\}_{1 \leq j \leq m}$   and energy arrival rate $\lambda_i$. The overall system state transition probability   can be factorized as: 

\footnotesize
\begin{eqnarray}\label{eqn:state_transition_prod}
\mathbb{P}(\bm{s}(t+1)| \bm{s}(t), \bm{b}(t), \bm{a}(t)) 
= \prod_{i=1}^N \mathbb{P}(s_{i}(t+1)| s_{i}(t), b_{i}(t), a_{i}(t))  
\end{eqnarray}
\normalsize

where, \eqref{eqn:state_transition_prod} follows from the fact that the state transition is independent across source nodes, given $\{b_i(t),a_i(t)\}_{1 \leq i \leq N}$.
 
Now we can write \eqref{eqn:main-problem} as follows:

 \begin{align}\label{eqn:main-problem-dis}
\footnotesize
 \min_{\bm{\pi}}  \lim_{T \rightarrow \infty}\frac{1}{TN} \sum_{t=1}^T \sum_{i=1}^N  \mathbb{E}_{\bm{\pi}} (c_{i}(s_{i}(t), b_{i}(t),a_{i}(t)))\nonumber\\
 s.t. \hspace{0.2 cm } \sum_{i=1}^N  b_{i}(t) \leq 1, \hspace{0.3 cm } \forall t \geq 1
 \normalsize
\end{align}

\subsection{Constrained MDP, Lagrange Multiplier and Whittle Index} \label{section:multiple-sensor-CMDP-stationary}
We relax the constraint in \eqref{eqn:main-problem-dis}    such that the hard scheduling  constraint at each time is replaced by a  soft constraint on the time-averaged mean number of channel probes:

\begin{align}\label{eqn:main-problem_CMDP}
\footnotesize
 \min_{\bm{\pi}}  \lim_{T \rightarrow \infty}\frac{1}{TN} \sum_{t=1}^T \sum_{i=1}^N  \mathbb{E}_{\bm{\pi}} (c_{i}(s_{i}(t), b_{i}(t),a_{i}(t)))\nonumber\\
 s.t. \hspace{0.2 cm } \frac{1}{TN} \sum_{t=1}^T\sum_{i=1}^N \mathbb{E}_{\bm{\pi}} [ b_{i}(t)] \leq  \frac{1}{N} \hspace{0.3 cm } 
 \normalsize
\end{align}

Next, we use the Lagrangian relaxation method with a Lagrange multiplier $\hat{\mu}$ to convert  the above CMDP into an unconstrained MDP and minimize the following Lagrangian:

\footnotesize
\begin{eqnarray}\label{eqn:Lagrange_CMDP_policy}
 \mathcal{L}(\bm{\pi},\hat{\mu})  &=& \lim_{T \rightarrow \infty}\frac{1}{TN} \sum_{t=1}^T \sum_{i=1}^N  \bigg( \mathbb{E}_{\bm{\pi}}  [c_{i}(s_{i}(t), b_{i}(t),a_{i}(t)]+\nonumber\\ 
 &&\hat{\mu}  \mathbb{E}_{\bm{\pi}}  [b_{i}(t)] \bigg) 
\end{eqnarray}
\normalsize

The Lagrangian dual function for any $\hat{\mu} \geq 0$ is  defined by $\mathcal{L}^*(\hat{\mu})= \min_{\bm{\pi}} \mathcal{L}(\bm{\pi},\hat{\mu})$ and is achieved by using the $\hat{\mu}$-optimal policy    $\bm{\pi}^*_{\hat{\mu}}= \arg \min_{\bm{\pi}}\mathcal{L}(\bm{\pi},\hat{\mu})$. Interestingly, 
this problem of obtaining an $\hat{\mu}$-optimal policy $\bm{\pi}^*_{\hat{\mu}}$ can be decoupled into $N$ sub-problems across all sources. Thus, the Lagrangian in \eqref{eqn:Lagrange_CMDP_policy} can be decomposed as:

\footnotesize
\begin{eqnarray}\label{eqn:Lagrange_alternate_CMDP_policy}
 \mathcal{L}(\bm{\pi},\hat{\mu})  &=&  \frac{1}{N}  \sum_{i=1}^N \mathcal{L}_i(\bm{\pi}_i,\hat{\mu}) 
\end{eqnarray}
\normalsize

where, $\mathcal{L}_i(\bm{\pi}_i,\hat{\mu})$ is the Lagrangian for $i$-th sub-problem of source $i$ and is given by:

\footnotesize
\begin{eqnarray}\label{eqn:Lagrange_persource_CMDP_policy}
  \mathcal{L}_i(\bm{\pi}_i,\hat{\mu}) &=& \lim_{T \rightarrow \infty}\frac{1}{T} \sum_{t=1}^T  \mathbb{E}_{\bm{\pi}_i} \bigg( c_{i}(s_{i}(t), b_{i}(t),a_{i}(t)]+\nonumber\\ 
 &&\hat{\mu}  \mathbb{E}_{\bm{\pi}_i}  [b_{i}(t)] \bigg),  \hspace{0.6 cm}  \forall i \in \{1, 2, \cdots, N\}
\end{eqnarray}
\normalsize

Here $\bm{\pi}_i$ is the policy used for the $i$-th source. Obviously, the sub-problem for $i$-th source is an unconstrained average cost MDP. This unconstrained average cost MDP is solved from the solution of an $\alpha$-discounted cost  MDP problem \cite[Section 4.1]{bertsekas2011dynamic} in the regime $\alpha \uparrow 1$.  

 %\begin{align}\label{eqn:main-problem-discounted}
%\footnotesize
% \min_{\bm{\pi}}  \lim_{T \rightarrow \infty}\frac{1}{TN} \sum_{t=1}^T \sum_{i=1}^N \alpha^t \mathbb{E}_{\bm{\pi}} (c_{i}(s_{i}(t), b_{i}(t),a_{i}(t)))\nonumber\\
% s.t. \hspace{0.2 cm } \sum_{i=1}^N  b_{i}(t) \leq 1 \hspace{0.3 cm } \forall t
% \normalsize
%\end{align}

\subsubsection{Bellman Equation}
Let $J_{i}^{*}(E_i,K_i)$ denote the optimal value function for state $(E_i,K_i)$ of source $i$ in the discounted cost problem, and let $W_{i}^*(E_i,K_i,C(i,\cdot))$ be the cost-to-go from an intermediate state $(E_i,K_i,C(i, \cdot))$, where $C(i,\cdot) \in \{C_1,C_2,\cdots,C_m\}$. The Bellman equations \cite[Proposition 7.2.1]{bertsekas2005dynamic}  for the  $\alpha$-discounted cost MDP problem for each $i$-th source are given by \eqref{eqn:Bellman-eqn-multiple-source-with-fading-general} which are solved by standard value iteration technique \cite[Proposition 7.3.1]{bertsekas2005dynamic}.

\scriptsize
\begin{eqnarray} \label{eqn:Bellman-eqn-multiple-source-with-fading-general}
J^{*}_{i}(E_{i} \geq E_{s},K_i)&=& \min \bigg\{u_i(E_i,K_i), v_i(E_i,K_i) \bigg\}\nonumber\\
u_i(E_i,K_i) &=& K_i+\alpha \mathbb{E}_{A_i}J_i^{*}(\min\{E_i+A_i,B_i\},K_i+1), \nonumber\\
v_{i}(E_i,K_i)&=& \hat{\mu}+\sum_{j=1}^m q_{j,i} W^{*}_{i}(E_i,K_{i},C_{j})\nonumber\\ 
W^{*}_{i}(E_i,K_i,C(i, \cdot))&=& \min\{K_i+\alpha \mathbb{E}_{A_i}J^{*}_{i}(\min\{E_i+A_i,B_i\},K_i+1),\nonumber\\ 
&&K_i(1-p(C(i, \cdot)))+\alpha p(C(i, \cdot))\mathbb{E}_{A_i}J^{*}_{i}(\min\{E_i\nonumber\\
&&-E_{s}+A_i,B_i\},1)+\alpha (1-p(C(i, \cdot)))\mathbb{E}_{A_i}J^{*}_{i}(\nonumber\\
&&\min\{E_i-E_{s}+A_i,B_i\},K_{i}+1)\}\nonumber\\
J^{*}_{i}(E_i< E_{s},K_i)&=& K_i+\alpha \mathbb{E}_{A_i}J^{*}_{i}(\min\{E_i+A_i,B_i\},K_i+1)
\end{eqnarray}
\normalsize

Here  the term $u_i(E_i, K_i)$ in \eqref{eqn:Bellman-eqn-multiple-source-with-fading-general} is the cost of not selecting the $i$-th source to probe channel state ($b_{i}(t)=0$), which comprises single-stage AoI cost $K_i$ and an $\alpha$ discounted future cost for a random next state $(\min\{E_i+A_i,B_i\},K_i+1)$,  averaged over $A_i$. The quantity $v_{i}(E_i, K_i)$ is the expected cost of selecting source $i$ to probe the channel state which includes the penalty cost $\hat{\mu}$ for choosing the $i$-th source for probing. After probing the $i$-th source, it encounters an intermediate state $(E_i,K_i,C(i, \cdot))$, for which if $a_{i}(t)=0$, a single stage sampling cost $K_i$ is induced and the next state becomes $(\min\{E_i+A_i,B_i\},K_i+1)$; and if $a_{i}(t)=1$, the expected sampling cost is $K_{i}(1-p(C(i, \cdot)))$ (where the expectation is calculated using the packet success probability $p(C(i, \cdot))$) of source $i$ for the observed channel state $C(i, \cdot)$, and the random next state becomes $(min\{E_{i}-E_s+A_i,B_i\},1)$ and $(min\{E_{i}-E_s+A_i,B_i\},K_{i}+1)$ if $r_{i}(t)=1$ and $r_{i}(t)=0$, respectively. The last equation in \eqref{eqn:Bellman-eqn-multiple-source-with-fading-general} follows similarly since when $E_i<E_s$, $(b_{i}(t)=0, a_{i}(t)=0)$ is the only feasible action due to lack of energy. Substituting the value of $u_{i}(E_i,K_i)$ and  $v_{i}(E_i,K_i)$ in the first equation of \eqref{eqn:Bellman-eqn-multiple-source-with-fading-general}, we get:

 \scriptsize
\begin{eqnarray}\label{eqn:Bellman-eqn-single-sensor-single-process-with-fading}
&&J^{*}_i(E_i \geq E_{s},K_i)\nonumber\\
&=&\min \bigg\{K_i+\alpha \mathbb{E}_{A_i}J^{*}_i(\min\{E_i+A_i,B_i\},K_i+1),\hat{\mu}+\mathbb{E}_{C(i, \cdot)} \bigg( \min\{\nonumber\\ 
&&K_i+\alpha \mathbb{E}_{A_i}J^{*}_{i}(\min\{E_i+A_i,B_i\},K_i+1),K_i(1-p(C(i, \cdot)))+\nonumber\\ 
&&\alpha p(C(i, \cdot))\mathbb{E}_{A_i}J^{*}_{i}(\min\{E_i-E_{s}+A_i,B_i\},1)+\nonumber\\
&&\alpha (1-p(C(i, \cdot)))\mathbb{E}_{A_i}J^{*}_{i}(\min\{E_i-E_{s}+A_i,B_i\},K_{i}+1)\} \bigg) \bigg\} \nonumber\\
&&J^{*}_i(E_i < E_{s},K_i)\nonumber\\
&=&K_i+\alpha \mathbb{E}_{A_i}J^{*}_i(min\{E_i+A_i,B_i\},K_i+1) 
\end{eqnarray}
\normalsize

For the $i$-th sub-problem with optimal policy $\pi_{i}^*$, let us denote by $\mathcal{I}(\hat{\mu})$   the set of states for which it is optimal to not probe the $i$-th source, when the source probing charge is $\hat{\mu}$,  i.e., $ \mathcal{I}(\hat{\mu})= \{(E_i, K_i) \in \{0,1, \cdots, B_i\} \times \mathbb{Z}^+: u_i(E_i, K_i)< v_i(E_i, K_i)\}$. 

\begin{definition}
    
 The sub-problem corresponding the $i$-th source is indexable, if $\mathcal{I}(\hat{\mu})$ increases monotonically form $\emptyset$ to the entire space as $\hat{\mu}$  increases from $0$ to $\infty$. The AoI minimization problem is indexable if all the $N$  sub-problems are indexable. 
\end{definition}

\begin{conjecture}\label{conjecture:multiple-source-WI}
 The indexability properly holds for all the sub-problems. Let us denote by $\mathcal{WI}_i(E_i, K_i)$   the Whittle index for state $(E_i, K_i)$ of source $i$ which is the infimum charge $\hat{\mu}$ that makes probing and not probing decisions of source $i$ equally desirable in state $(E_i, K_i)$. 
\end{conjecture}

\begin{remark}
    Some papers  \cite{kadota2016minimizing, sun2019closed, tong2022age, kriouile2021global, tripathi2019whittle}, in the AoI literature have proved indexability and used Whittle's index for source selection. However, for EH sources, the state of each node~$i$ is a tuple $(E_i, K_i)$, which prohibits the direct application of the techniques used in  \cite{kadota2016minimizing,  kriouile2021global}, to prove indexability in case if one-dimensional state.
\end{remark}

\subsubsection{Policy structure}\label{subsubsection:Policy structure-IID-System-N=1} \hfill\\\\
\noindent {\bf Probing Policy: }
The scheduler first computes the Whittle's index $\mathcal{WI}_i(E_i, K_i)$ for each source $i$, and then selects the source with the highest Whittle index for channel probing.

\begin{conjecture}\label{conjecture:single-sensor-single-process-with-fading-policy-structure}
 At any time $t$, the optimal policy for the $\alpha$-discounted  cost problem is to select for channel probing the $i^*$-th source node with the highest Whittle's index such that $E_{i^*} \geq E_s$, i.e., $ i^*(t)=\arg \max_{1 \leq i \leq N} \mathcal{WI}_i(t)$. 
\end{conjecture}

\noindent {\bf Sampling Policy: }
For deriving the optimal sampling policy of the $i$-th source, we need to prove the following properties of cost functions which can be easily done by following similar steps as in \cite[Appendix B]{jaiswal2024age}.

\begin{lemma} \label{lemma:single-sensor-single-process-J-decreasing-in-p}
For each source $i$,  
$J^{*}_{i}(E_{i}, K_{i})$ is increasing in $K_{i}$ and $W^{*}_{i}(E_{i},K_{i},C(i, \cdot)$ is decreasing in $p(C(i,\cdot))$. %\textcolor{red}{write mathematically}????????????????????
\end{lemma}

\begin{theorem}\label{theorem:single-sensor-single-process-with-fading-policy-structure}
 At any time $t$, after selecting the $i^*$-th source node to probe the channel,  the optimal sampling policy for the selected source is a threshold policy on $p(C (i^*,\cdot))$. Node $i^*$ has to be sampled if $p(C(i^*,\cdot))\geq p_{th}(E_{i^*}, K_{i^*})$ for a threshold   $p_{th}(\cdot,\cdot)$. 
\end{theorem}

\begin{proof} Similar to that of Theorem 1 in  \cite[Appendix C]{jaiswal2024age}).
\end{proof}

The following intuitions are supported by policy structure. Firstly, given the energy and age of all the source nodes, we select the $i^*$-th source node with the highest Whittle's index to probe the channel state. 
Next, given $E_{i^*}$, $K_{i^*}$, and probed channel state $C(i^*,\cdot)$, if the packet success probability is above a threshold, then it is optimal to sample $i^*$-th source and transmit the measurement packet.% We will later numerically observe in Section~\ref{section:numerical-work} some intuitive properties of $p_{th}(E_i,K_i)$ as a function of $E_i$, $K_i$ and $\lambda_i$.\\

Our proposed WITS3 policy is summarised in  Algorithm~\ref{alg:cap}.

\begin{algorithm}
\caption{ Whittle's Index and Threshold based Source Scheduling and Sampling  (WITS3) policy }\label{alg:cap}
\begin{algorithmic}
\State $\bm{Input}: \mathcal{S},\mathcal{A}, N, \{\lambda_i, q_{j,i}, B_i, K_{i}\}_{ 1 \leq i \leq N, 1 \leq j \leq m } $.
%\State $\bm{Initialize}$: $\hat{\lambda}(0)=0$,  $t\leftarrow 1$,  initial state   $s_{1}$.
\For{  $ t=1, 2,3 \cdots$}
\State{ $\forall i \in \{1,2,\cdots,N\}$ such that $E_i \geq E_s$}
\State $\bm{Source Probing Policy}$:
\State Calculate the Whittle's index of each source. 
\State Select the $i^*$-th source having the highest Whittle index to probe the channel state ($b_{i^*}(t)=1$), where\\ $ i^*=\arg \max_{i} \mathcal{WI}_i$.
\State Observe the current channel state $C(i^*,t)$ (and thus  the corresponding packet success probability $p(C(i^*,t))$ or $p(i^*,t)$) for the $i^*$-th source. 
\State $\bm{Source Sampling Policy}$:
\If{ $p(C(i^*,t)) \geq p_{th}(E_{i^*}(t), K_{i^*}(t))$  (where $p_{th}(\cdot,\cdot)$ is obtained via value iteration)} 
\State Sample the $i^*$-th source and send the packet ($a_{i^*}(t)=1$),  observe $r_{i^*}(t)$.
\Else{ $a_{i^*}(t)=0$}.
\EndIf
\State{ $ \forall  i \in \{1,2,\cdots,N\}$, such that  $i \neq i^*$  or  $E_i < E_s$}
\State $ b_{i}(t)=0$ (no  probing) and $a_{i}(t)=0$ (no sampling). 
\EndFor
\end{algorithmic}
\end{algorithm}

 \vspace{-6pt}
\section{Q-Learning for an unknown environment}\label{section:RL}
In this section, we assume that the channel statistics and the energy harvesting characteristics are not known, and need to be learnt from the history of states, actions, and observations, in order to solve   \eqref{eqn:main-problem}. Specifically, we assume  that the energy generation rate $\lambda_i$ and the channel state probabilities $\{q_{j, i}\}_{1 \leq j \leq m}$ associated with the  $i$-th source node, $\forall i \in \{1,2,..\cdots, N\}$, are unknown. Also, the packet success probability $p(i,t)$ for $i$-th source node at time $t$ may be known or unknown. Our algorithm Q-WITS3 can handle both cases.  
We use the Q-learning technique \cite[Section 11.4.2]{bhatnagar2013reinforcement}   to learn Q-values and Whittle's indices for the unconstrained MDP subproblems.  \\
 
  \vspace{-10pt}
 \subsubsection{Optimality equation for decoupled subproblems expressed as Q-function}
 \label{section:Optimality equation for Q-fuction SSSP}

Recall from Section~\ref{section:multiple-sensor-CMDP-stationary} that, in order to solve   \eqref{eqn:main-problem},  we use a Lagrangian multiplier $\hat{\mu}$  to convert the CMDP to an unconstrained MDP and further decouple this optimization problem across sources. For a generic state-action pair $(E_i,K_i,b_i)$,  the optimal Q-function corresponding to the $i$-the source node is represented by $Q^{*}_i(E_i,K_i,b_i)$, where $(E_i,K_i) $ denotes the  state and the action is given by $b_i \in \{0,1\}$. Here  $Q^{*}_i(E_i,K_i,b_i)$ indicates the expected cost incurred if the current state and action are   $(E_i,K_i)$ and $b_i$ respectively, and an optimal policy is implemented from the next time instant. Furthermore, for an intermediate state $(E_i,K_i,C(i,\cdot))$ and   action $a_i \in \{0,1\}$, the optimal Q-function is represented by $Q^{*}_i(E_i,K_i,C(i,\cdot), a_i)$. It is crucial to acknowledge that in order to manage primary states and intermediate states, two distinct types of Q-values must be maintained for each source $i$; this is not the case with conventional Q-learning.  

 The optimal Q-value   for $i$-th decoupled source is given below:

 \scriptsize
\begin{eqnarray}\label{eqn:Q-eqn-single-process-1} 
&&Q^{*}_i(E_i \geq E_{s}, K_i, b_i=0) \nonumber\\
&=&K_i+\alpha \mathbb{E}_{A_i}J^{*}(min\{E_i+A_i,B_i\},K_i+1) \nonumber \\
&=& K_i+\alpha \mathbb{E}_{A_i}\min_{b_i^{'} \in \{ 0,1 \}}Q^{*}(min\{E_i+A_i,B_i\}, K_i+1,b_i^{'})
\end{eqnarray}
\normalsize

In \eqref{eqn:Q-eqn-single-process-1}, the terms on the right-hand side of the first equality operator include the immediate cost of not selecting the $i$-th source node  to probe the channel state by scheduler $(b_i=0)$ when state is $(E_i \geq E_{s}, K_i)$ and an $\alpha$ discounted future cost $\mathbb{E}_{A_i}J^{*}_i(\min\{E_i+A_i,B_i\},K_i+1)$  that is later replaced by 
$\mathbb{E}_{A_i}\min_{b_i^{'} \in \{ 0,1 \}}Q^{*}_i(\min\{E_i+A_i,B_i\},K_i+1,b_i^{'})$ (averaged over  $A_i$). However, 

\scriptsize
\begin{eqnarray} \label{eqn:Q-eqn-single-process-2}
Q^{*}_i(E_i \geq E_{s}, K_i,b_i=1)   
&=& \hat{\mu}+\sum_{j=1}^m q_{j,i} W_i^{*}(E_i,K_i,C_{j})\nonumber\\ 
&=& \hat{\mu}+\sum_{j=1}^m q_{j,i}\min_{a_i \in \{ 0,1 \}}Q^{*}_i(E_i,K_i,C_{j}, a_i) \nonumber\\ 
\end{eqnarray}
\normalsize

where $\hat{\mu}$ in the right-hand side of \eqref{eqn:Q-eqn-single-process-2} is the penalty for selecting the $i$-th source node to probe the channel state. Based on the probed channel quality, the selected source will take sampling decision $a_i \in \{ 0,1 \}$ for an intermediate state $(E_i,K_i, C(i,\cdot))$. Thus, 

\scriptsize
\begin{eqnarray} \label{eqn:Q-eqn-single-process-4}
&&Q^{*}_i(E_i,K_i, C(i,\cdot),a_i=0)\nonumber\\
&=&K_i+\alpha \mathbb{E}_{A_i}\underbrace{\min_{b_i^{'} \in \{ 0,1 \}}Q^{*}_i(\min\{E_i+A_i,B_i\}, K_i+1,b_i^{'})}_{=J^*_i(\min\{E_i+A_i,B_i\}, K_i+1)}
\end{eqnarray}
\normalsize

and 

\scriptsize
\begin{eqnarray} \label{eqn:Q-eqn-single-process-5}
&&Q^{*}_i(E_i,K_i,C(i,\cdot),a_i=1)\nonumber\\
&=&K_i(1-p(C(i, \cdot)))+\alpha p(C(i, \cdot))\mathbb{E}_{A_i}\min_{b_i^{'} \in \{ 0,1 \}}Q^{*}_i(\min\{E_i-E_{s}+\nonumber\\
&&A_i,B_i\},1,b_i^{'})+\alpha (1-p(C(i, \cdot))\mathbb{E}_{A_i}\min_{b_i^{'} \in \{ 0,1 \}}Q^{*}_i(\min\{E_i-E_{s}+\nonumber\\
&&A_i,B_i\},K_i+1,b_i^{'})
\end{eqnarray}
\normalsize

Likewise, \eqref{eqn:Q-eqn-single-sensor-single-process-with-fading-general} represents  the optimal Q-function for state $(E_i< E_{s},K_i)$ where the only practical action is   $b_i = 0$ due to lack of energy. Thus, 

\scriptsize
\begin{eqnarray} \label{eqn:Q-eqn-single-sensor-single-process-with-fading-general}
&&Q^{*}_i(E_i< E_{s},K_i,b_i=0)\nonumber\\
&=& K_i+\alpha \mathbb{E}_{A_i}J^{*}_i(\min\{E_i+A_i,B_i\},K_i+1)\nonumber\\
&=&K_i+\alpha \mathbb{E}_{A_i}\min_{b_i^{'} \in \{ 0,1 \}}Q^{*}_i(\min\{E_i+A_i,B_i\},K_i+1,b_i^{'})
\end{eqnarray}
\normalsize

\subsubsection{Q-Learning algorithm for decoupled subproblems}\label{section:Optimality equation for Q-fuction SSSP}
Let us recall that Whittle's index of source $i$ in state $(E_i, K_i)$ is defined as the value $\hat{\mu}_i(E_i, K_i)$ of $\hat{\mu}_i$ for which both the actions of probing and not probing the $i$-th source are equally preferable in state $(E_i, K_i)$. With some abuse of notation, let us define $Q^{\hat{\mu}_i}_i(E_i,K_i,b_i) \doteq Q_{i}(E_i,K_i,\hat{\mu}_i,b_i)$ $\forall i \in \{1,2,\cdots,N\}$.  Then, if $\hat{\mu}_i(E_i,K_i)$ or simply $\hat{\mu}_i$ is the Whittle's index for state $(E_i,K_i)$, then:

\footnotesize
\begin{eqnarray}\label{mu-learning-criterion}
Q_{i}(E_i,K_i,\hat{\mu}_i,b_i=1)=Q_{i}(E_i,K_i, \hat{\mu}_i,b_i=0) 
\end{eqnarray}
\normalsize

Our goal is to learn all Q-values and Whittle's indices.
To this end, we employ two-time scale stochastic approximation iterations \cite{borkar2009stochastic}.  Q-learning is carried out in a faster timescale, where $\hat{\mu}_i$ remains quasi-static. On the other hand,  $\hat{\mu}_i(E_i,K_i)$ for all nodes and for all states are updated at a slower timescale where the Q-values appear to have converged.

The optimal Q-functions for the $i$-th decoupled source, $\forall i \in \{1,2,\cdots,N\}$ are given by the zeros of \eqref{eqn:Q-eqn-single-process-1}-\eqref{eqn:Q-eqn-single-sensor-single-process-with-fading-general}.  Since all states are not encountered at a given time $t$, we employ  asynchronous stochastic approximation \cite{borkar1998asynchronous} to iteratively find the optimal Q-functions and learn Whittle's indices for all the $N$ sources, beginning from any arbitrary Q-functions for each source.  Two-time scale  asynchronous stochastic approximation iterations involve two step size sequences:   $\{d(t)\}_{t \geq 0}$ for learning the Q-values at a faster timescale,  and $\{f(t)\}_{t \geq 0}$ for learning Whittle's indices at a slower timescale. We make the following assumptions \cite{bhatnagar2011borkar, borkar2009stochastic}.
\begin{assumption}\label{assumption2}
(i) $0< d(t) \leq \bar{d}$ where $\bar{d}>0$, \\
(ii) $d(t+1)\leq d(t)$ for all $t \geq 0$, \\
(iii) $\sum_{t=0}^\infty d(t) = \infty$ and $\sum_{t=0}^\infty (d(t))^{2} < \infty$.\\
(iv) $0< f(t) \leq \bar{f}$ where $\bar{f} > 0$, \\
(v) $f(t+1)\leq f(t)$ for all $t \geq 0$, \\
(vi) $\sum_{t=0}^\infty f(t) = \infty$ and $\sum_{t=0}^\infty (f(t))^{2} < \infty$.\\
(vii) $f(t) = o(d(t))$
\end{assumption}

Let $\nu_{t}(s_i,b_i)$ indicate the number of occurrences of the primary state-action pair $(s_i,b_i)$ of source $i$ up to iteration $t$, where $s_i \in \mathcal{S}_i $ and $b_i \in \{0,1\}$. A comparable notation is assumed for any generic intermediate state $v_i$ and its corresponding action $a_i$ of source $i$. The convergence of asynchronous stochastic approximations requires the following assumption \cite{bhatnagar2011borkar}:
\begin{assumption}\label{assumption1}
   $\lim \inf_{t \rightarrow \infty} \frac{\nu_{t}(s_i,b_i)}{t} > 0$ almost surely $\forall (s_i,b_i)$, and $\lim \inf_{t \rightarrow \infty} \frac{\nu_{t}(v_i,a_i)}{t} > 0$ almost surely $\forall (v_i,a_i)$ and for all source $i \in \{1,2,\cdots,N\}$.
 \end{assumption}

The proposed Q-learning algorithm iteratively modifies the entries of a look-up table $Q_{i,t}(\cdot,\cdot)$ of each source $i$ for different state-action pairs based on the current state, the action taken, and the observed next state of that particular source.
 However, to satisfy Assumption~\ref{assumption1}, it is necessary to visit all state-action pairs of each source $i$ infinitely and comparatively often. This is guaranteed by considering a small exploration probability $\epsilon \in (0,1)$; at each decision instant, with probability $\epsilon$, the scheduler chooses one of the sources uniformly at random,  and, with probability $(1-\epsilon)$, it selects the source with the highest Whittle's index to probe the channel state. This ensures visitation to all states and intermediate states infinitely and comparatively often. Similarly, while making the sampling decision for the selected source $i^*$ with observed channel state $C(i^*, \cdot)$,  it picks an action uniformly at random with probability $\epsilon$, and takes action $\arg \min_{a_{i^*}} Q_i(v_{i^*}, a_{i^*})$ with probability $(1-\epsilon)$, where $v_{i^*}$ is the corresponding intermediate state.

The two timescale iterations  corresponding to the source $i$ for various states and intermediate states   are as follows:

\scriptsize
\begin{eqnarray} \label{eqn:Q-iteration-single-process-iterate1}
&&Q_{i,t+1}(E_i \geq E_{s},K_i, \hat{\mu}_{i,t}(E_i, K_i),b_i=0)\nonumber\\
&=&Q_{i,t}(E_i,K_i, \hat{\mu}_{i,t}(E_i, K_i),b_i=0)+d(\nu_{t}(E_i,K_i,b_i=0))\mathrm{\mathbf{1}}\{s_i(t)= \nonumber\\
&&(E_i,K_i),b_i(t)=0\}[K_i+\alpha \min_{b_i^{'} \in \{ 0,1 \}}Q_{i,t}(\min\{E_i+A_i(t),B_i\},\nonumber\\
&&K_i+1,\hat{\mu}_{i,t}(E_i(t+1), K_i(t+1)),b_i^{'})-Q_{i,t}(E_i,K_i, \hat{\mu}_{i,t}(E_i, K_i),\nonumber\\
&&b_i=0)]
\end{eqnarray}
\normalsize

\scriptsize
\begin{eqnarray} \label{eqn:Q-iteration-single-process-iterate2}
&&Q_{i,t+1}(E_i \geq E_{s},K_i, \hat{\mu}_{i,t}(E_i, K_i),b_i=1) \nonumber\\
&=& \hat{\mu}_{i,t}(E_i, K_i)+ Q_{i,t}(E_i,K_i, \hat{\mu}_{i,t}(E_i, K_i),b_i=1)+d(\nu_{t}(E_i,K_i,\nonumber\\
&&b_i=1)\mathrm{\mathbf{1}}\{s_i(t)=(E_i,K_i),b_i(t)=1\}[\min_{a_i\in \{ 0,1 \}}Q_{i,t}(E_i,K_i,C(i,t), \nonumber\\
&&\hat{\mu}_{i,t}(E_i, K_i), a_i)-Q_{i,t}(E_i,K_i,\hat{\mu}_{i,t}(E_i, K_i),b_i=1)]
\end{eqnarray}
\normalsize

\scriptsize
\begin{eqnarray} \label{eqn:Q-iteration-single-process-iterate4}
&&Q_{i,t+1}(E_i,K_i,C(i,\cdot),\hat{\mu}_{i,t}(E_i, K_i),a_i=0)\nonumber\\
&=&Q_{i,t}(E_i,K_i,C(i,\cdot),\hat{\mu}_{i,t}(E_i, K_i),a_i=0)+d(\nu_{t}(E_i,K_i,C(i,\cdot),\nonumber\\
&&a_i=0))\mathrm{\mathbf{1}}\{v_i(t)=(E_i,K_i,C(i,\cdot)),a_i(t)=0\}[K_i+\alpha \min_{b_i^{'}\in \{ 0,1 \}}Q_{i,t}(\nonumber\\
&&\min\{E_i+A_i(t),B_i\},K_i+1,\hat{\mu}_{i,t}(E_i(t+1), K_i(t+1)),b_i^{'})-\nonumber\\
&&Q_{i,t}(E_i,K_i,C(i,\cdot),\hat{\mu}_{i,t}(E_i, K_i),a_i=0)]
\end{eqnarray}
\normalsize

\scriptsize
\begin{eqnarray} \label{eqn:Q-iteration-single-process-iterate4.1}
&&Q_{i,t+1}(E_i,K_i,C(i,\cdot),\hat{\mu}_{i,t}(E_i, K_i),a_i=1)\nonumber\\
&=&Q_{i,t}(E_i,K_i,C(i,\cdot),\hat{\mu}_{i,t}(E_i, K_i),a_i=1)+d(\nu_{t}(E_i,K_i,C(i,\cdot),\nonumber\\
&&a_i=1))\mathrm{\mathbf{1}}\{v_i(t)=(E_i,K_i,C(i,\cdot)),a_i(t)=1\}[K_i(1-p(C(i,\cdot)))+\nonumber\\
&&\alpha p(C(i,\cdot))\min_{b_i^{'}\in \{ 0,1 \}}Q_{i,t}(\min\{E_i-E_{s}+A_i(t),B_i\},1,\hat{\mu}_{i,t}(E_i(t+1),\nonumber\\
&& K_i(t+1)),b_i^{'})+\alpha (1-p(C(i,\cdot)))\min_{b_i^{'}\in \{ 0,1 \}}Q_{i,t}(\min\{E_i-E_{s}+A_i(t),\nonumber\\
&& B_i\},K_i +1,\hat{\mu}_{i,t}(E_i(t+1), K_i(t+1)),b_i^{'})-Q_{i,t}(E_i,K_i,C(i,\cdot),\nonumber\\
&&\hat{\mu}_{i,t}(E_i, K_i),a_i=1)]
\end{eqnarray}
\normalsize

\scriptsize
\begin{eqnarray} \label{eqn:Q-iteration-single-process-iterate5}
&&Q_{i,t+1}(E_i<E_{s},K_i,\hat{\mu}_{i,t}(E_i, K_i),b_i=0)\nonumber\\
&=&Q_{i,t}(E_i,K_i,\hat{\mu}_{i,t}(E_i, K_i),b_i=0)+d(\nu_{t}(E_i,K_i,b_i=0))\mathrm{\mathbf{1}}\{s_i(t)=\nonumber\\
&&(E_i,K_i),b_i(t)=0\}[K_i+ \alpha \min_{b_i^{'}\in \{ 0,1 \}}Q_{i,t}(\min\{E_i+A_i(t),B_i\},\nonumber\\
&&K_i+1,\hat{\mu}_{i,t}(E_i(t+1), K_i(t+1)),b_i^{'})-Q_{i,t}(E_i,K_i,\hat{\mu}_{i,t}(E_i, K_i),\nonumber\\
&&b_i=0)]
\end{eqnarray}
\normalsize

In these iterations, we assume that the packet success probability $p(C(i,\cdot))$ associated with channel state $C(i,\cdot)$ is known. However, if these probabilities are unknown, then one can replace $p(C(i,\cdot))$   by the packet success indicator $r_i(t)$, and run the same iteration.

The Whittle indices $\hat{\mu}_{i}(E_i, K_i)$ for each  state $(E_i, K_i)$ for each node~$i$, are updated at a slower timescale:

\scriptsize
\begin{eqnarray} 
\hat{\mu}_{i,t+1}(E_i, K_i) &=& \hat{\mu}_{i,t}(E_i, K_i) +f(\nu_{t}(E_i,K_i))\mathrm{\mathbf{1}}\{s_i(t)=(E_i,K_i)\} \nonumber\\
&&  \bigg( Q_{i,t}(E_i,K_i,\hat{\mu}_{i,t}(E_i, K_i),b_i=0) - Q_{i,t}(E_i,K_i, \nonumber\\
&&\hat{\mu}_{i,t}(E_i, K_i),b_i=1) \bigg)
\end{eqnarray}
\normalsize

It is important to note that, whenever the state $(E_i, K_i)$ is encountered, the {\em penalty} $\hat{\mu}_i (E_i, K_i)$ for probing node~$i$, is increased if $Q_{i,t}(E_i,K_i,\hat{\mu}_{i,t}(E_i, K_i), b_i=0) \geq Q_{i,t}(E_i,K_i,\hat{\mu}_{i,t}(E_i, K_i),b_i=1)$, and is decreased otherwise. In other words, this {\em penalty} is increased if the cost of not probing is more than the cost of probing.

{\bf Q-WITS3 algorithm: } In the Q-WITS3 algorithm, source probing decisions are made by using the latest available Whittle indices, subject to randomization with a small probability $\epsilon>0$ as discussed earlier. Under a generic intermediate state $(E_i,K_i, C(i,\cdot))$ at time $t$, action $\arg \min_{a_i \in \{0,1\}}Q_{i,t}(E_i,K_i, C(i,\cdot),\hat{\mu}_{i,t}(E_i, K_i),a_i)$is taken, again subject to randomization with a small probability $\epsilon>0$.

\section{Simulation results }\label{section:Simulation results}

\subsection{Multiple sources with known system parameters}\label{subsection:MS_KSP_Simulation results}

We consider three sources ($N=3$) with i.i.d. Bernoulli energy arrival processes   with rates $\lambda_{1}=0.6$, $\lambda_{2}=0.5$, and $\lambda_{3}= 0.4$. Each source $i$ is equipped with finite energy buffer size $B_i=5$ and incurs sampling energy cost   $E_{s}=1$ unit. The maximum age of a process of source $i$ is assumed to be upper bounded at $K_{i_{max}}=10$. Also, we assume that a fading channel has four states ($m=4$) and the channel state distribution $\{q_{j,i}\}_{1 \leq j \leq m, 1 \leq i \leq N}$ across the three sources is given by the matrix $\begin{bmatrix}
0.4 & 0.4 & 0.1 & 0.1\\
0.25 & 0.25 & 0.25 & 0.25\\
0.1 & 0.1 & 0.4 & 0.4\\
\end{bmatrix}^{T}$ and the corresponding packet success probabilities for all four channel states are given by $[0.9, 0.5, 0.3, 0.1]$.  Numerical exploration reveals that for each source $i$, the Whittle's index $(\mathcal{WI}_i)$ increases in both $E_i$ and $K_i$;  see Fig.~\ref{Whittle_index_inc}. This matches our intuition that a node with a higher age and more energy will get priority in scheduling. It is observed that the $\mathcal{WI}_1 \geq \mathcal{WI}_2 \geq \mathcal{WI}_3$ since source $1$ has a higher energy arrival rate and statistically better channel. 

Also,   Fig.~\ref{Whittel_pthesh} shows that the sampling threshold $p_{th}(E_i, K_i)$ decreases with $E_i$   since higher energy in the buffer allows the EH source to sample a process more aggressively. We also notice that, for fixed $K_i$ and $\hat{\mu}$ (the penalty for probing source $i$),   $p_{th}(\cdot,\cdot)$ is smallest for source~$1$ and largest for source $3$ due to the same reason as above.    Fig.~\ref{Whittel_pthesh}(a)-(b)  also show that   $p_{th}(\cdot,\cdot)$ decreases with $\hat{\mu}$; the intuition is that if a penalty for probing increases, then it is often better to sample and transmit more aggressively in order to avoid wastage of that time slot. Similarly,   $p_{th}(E_i, K_i)$ decreases with $K_i$, which we omit  due to space constraints.

Finally, we examine the effectiveness of the proposed policy by comparing against two other baseline policies: (i) Greedy Maximum Age with retransmission (GMA-R) allowed and (ii) Greedy maximum energy with retransmission (GME-R) allowed.   Fig.~\ref{Comp_plot_wae} shows that our proposed WITS3 algorithm has a lower cumulative average age (averaged over sources as well as over multiple sample runs) as compared to the greedy policies GMA-R and GME-R which probe the channel for the source having the largest age and energy respectively, and sample the selected source until either transmission succeeds or the source energy is completely depleted.

\begin{figure}[htb!]
  \begin{center}
 \subfloat[]{\includegraphics[height=3.8cm,width=8cm]{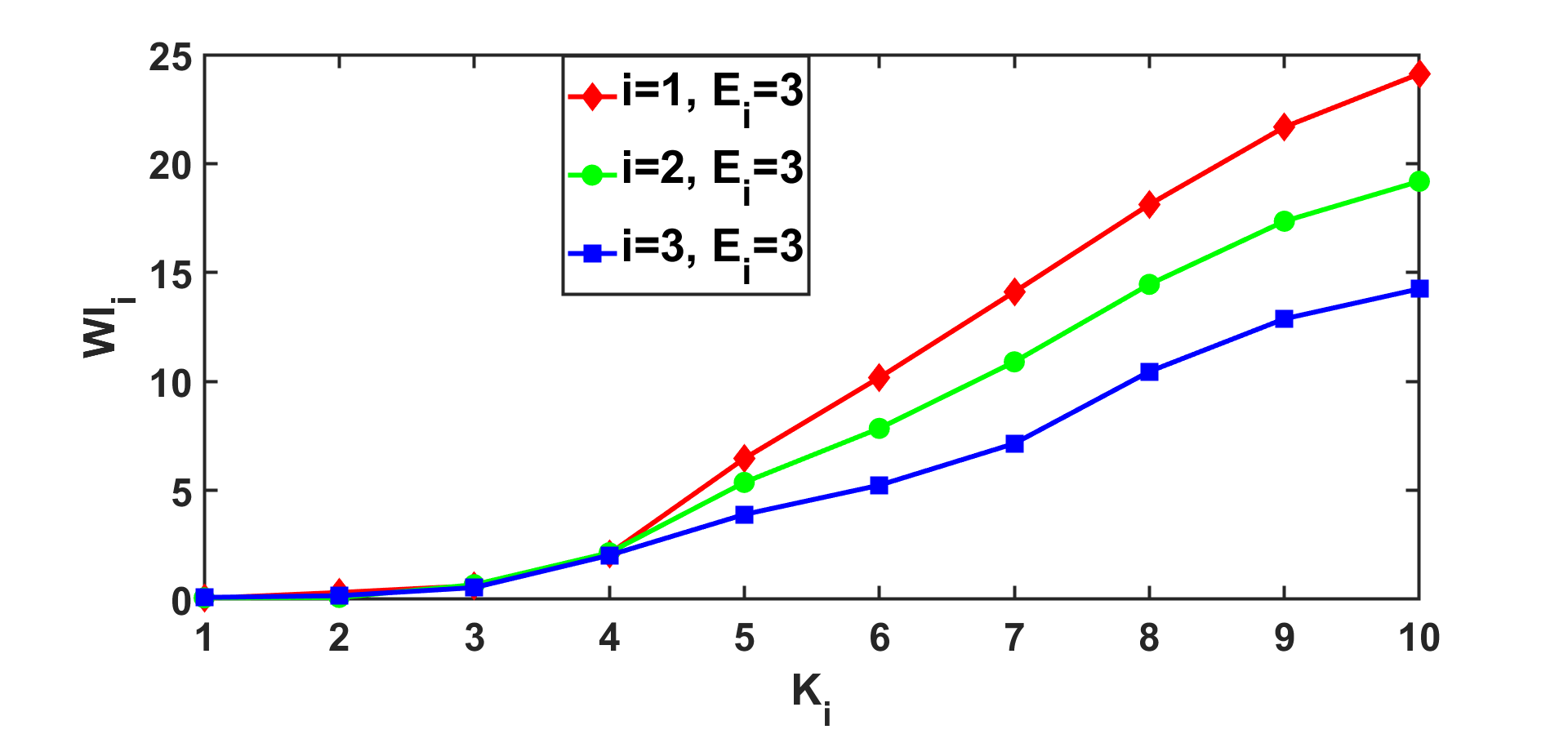}}
  % \vspace{-2pt}
   \hfill
  \subfloat[]{\includegraphics[height=3.8cm,width=8cm]{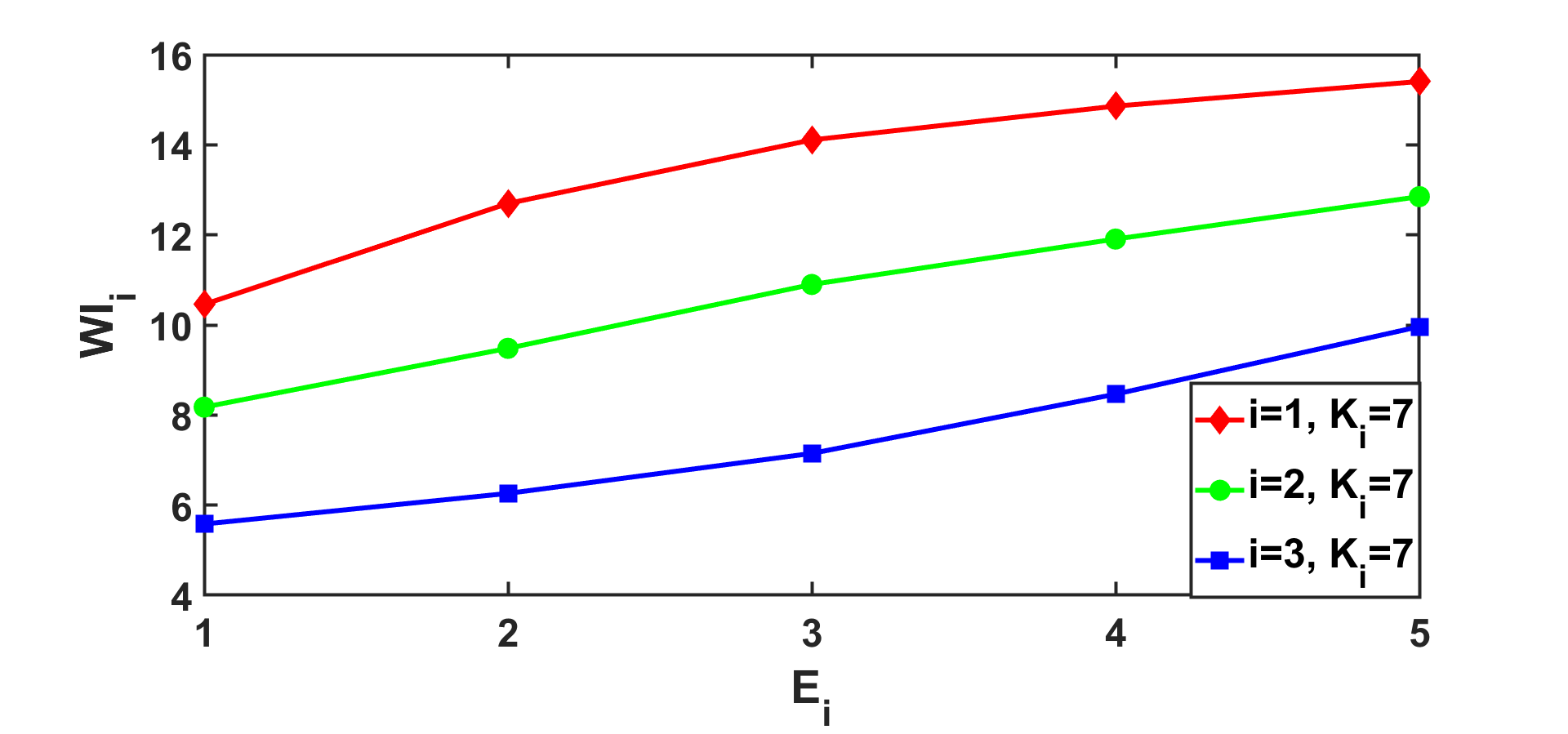}}
 \caption{Multiple sources ($N=3$): (a) Variation of $\mathcal{WI}_i$ with $K_i$ for a fixed $E_i$ and (b) Variation of $\mathcal{WI}_i$ with $E_i$ for a fixed $K_i$ where $i \in \{1,2,3\}$.}
 \label{Whittle_index_inc}
 \end{center}
 \vspace{-0.5cm}
 \end{figure}

\begin{figure}[t!]
  \begin{center}
 \subfloat[]{\includegraphics[height=3.8cm,width=8cm]{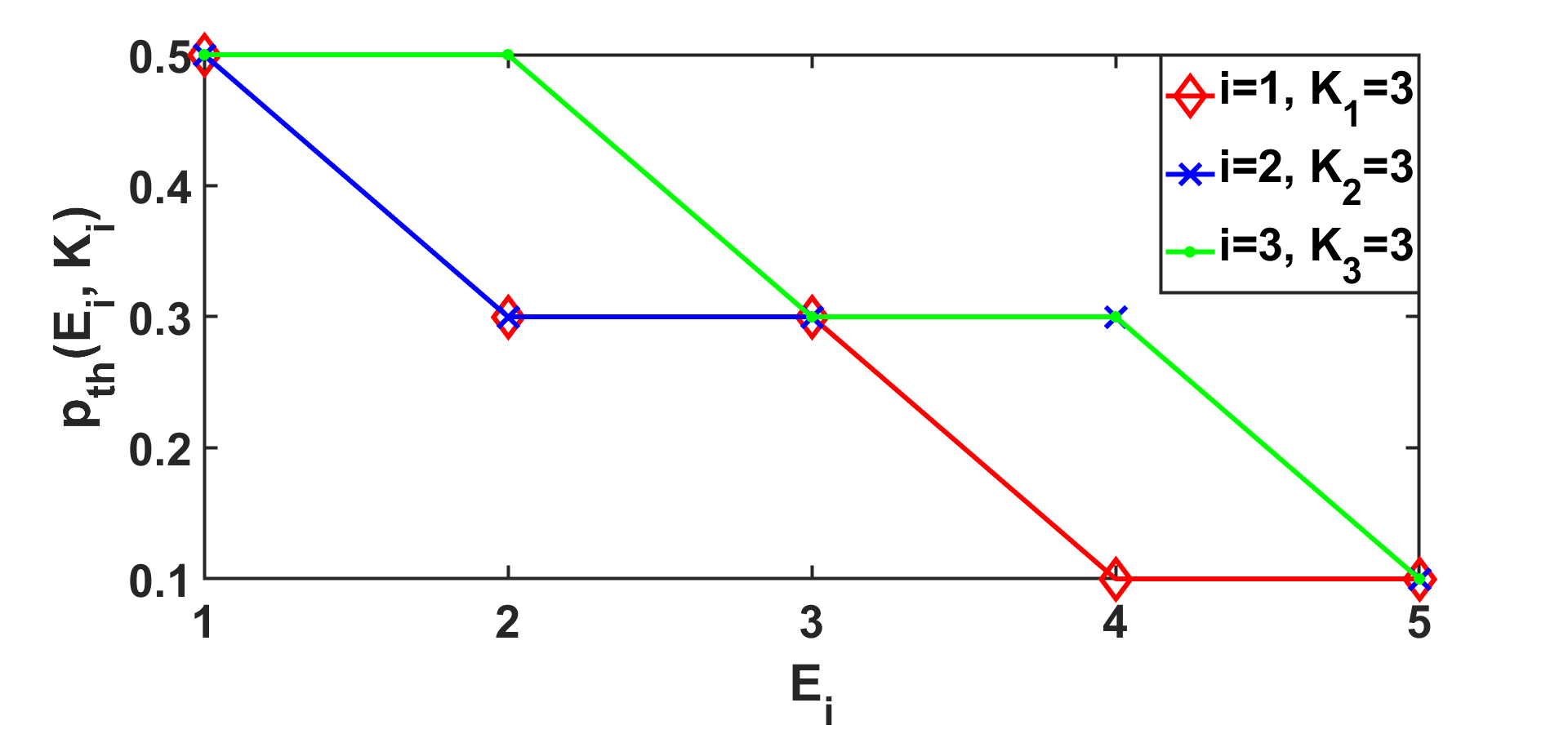}}
 %\vspace{-2pt}
 \hfill
  \subfloat[]{\includegraphics[height=3.8cm,width=8cm]{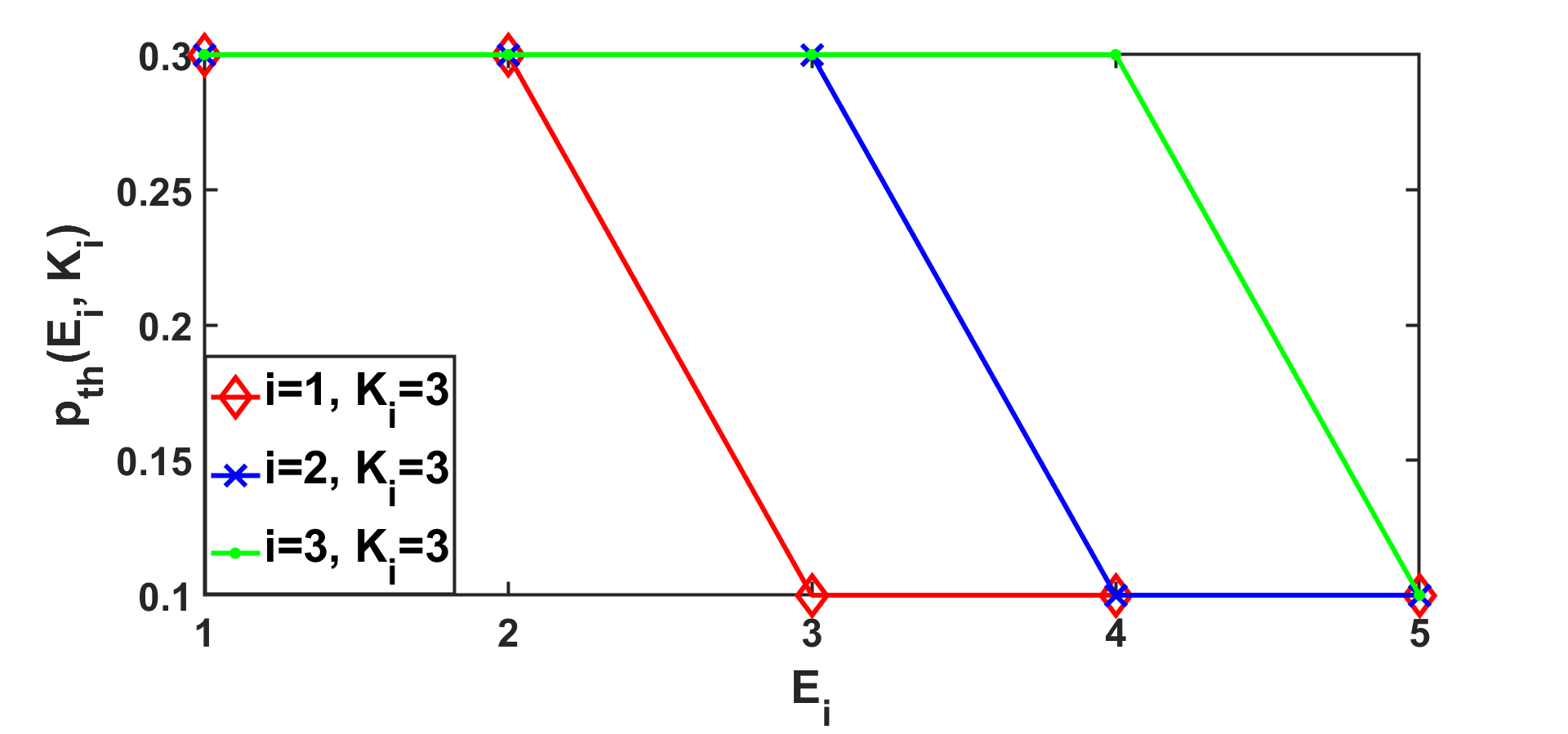}}
 \caption{Variation of $p_{th}(E_i,K_i)$ with $E_i$: (a)  $\hat{\mu}= 2 $, (b)   $\hat{\mu}= 4 $.  }
 \label{Whittel_pthesh}
 \end{center}
 \vspace{-15pt}
 \end{figure}

\begin{figure}[t!]
  \begin{center}
 \includegraphics[height=3.8cm,width=8cm]{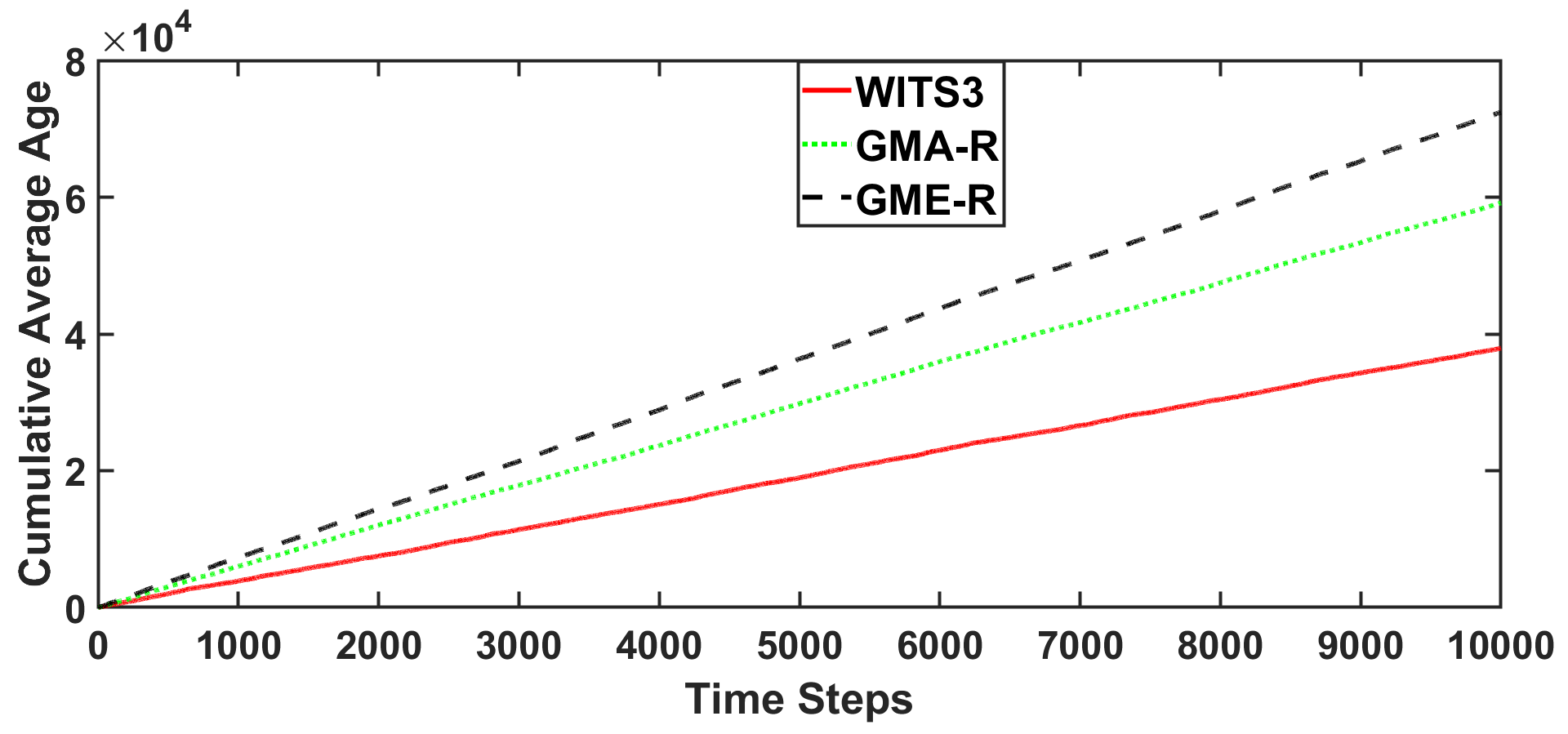}
 \caption{Comparison among WITS3, GMA-R and GME-R.}
 \label{Comp_plot_wae}
 \end{center}
\vspace{-18pt}
\end{figure}

\subsection{Q-WITS3 for unknown system parameters}\label{subsection:MS_USP_Simulation results}
Under the same system parameters as in Section~\ref{subsection:MS_KSP_Simulation results}, we evaluate the performance of Q-WITS3  averaged over five   sample paths, and compare against  the WITS3 policy. It is observed from Fig.~\ref{multi_avg_age_comp_plot} that the time-averaged AoI of   Q-WITS3   converges to that of WITS3, i.e., Q-WITS3 learns the WITS3 policy efficiently. 

\begin{figure}[t!]
  \begin{center}
 \includegraphics[height=4cm,width=8.9cm]{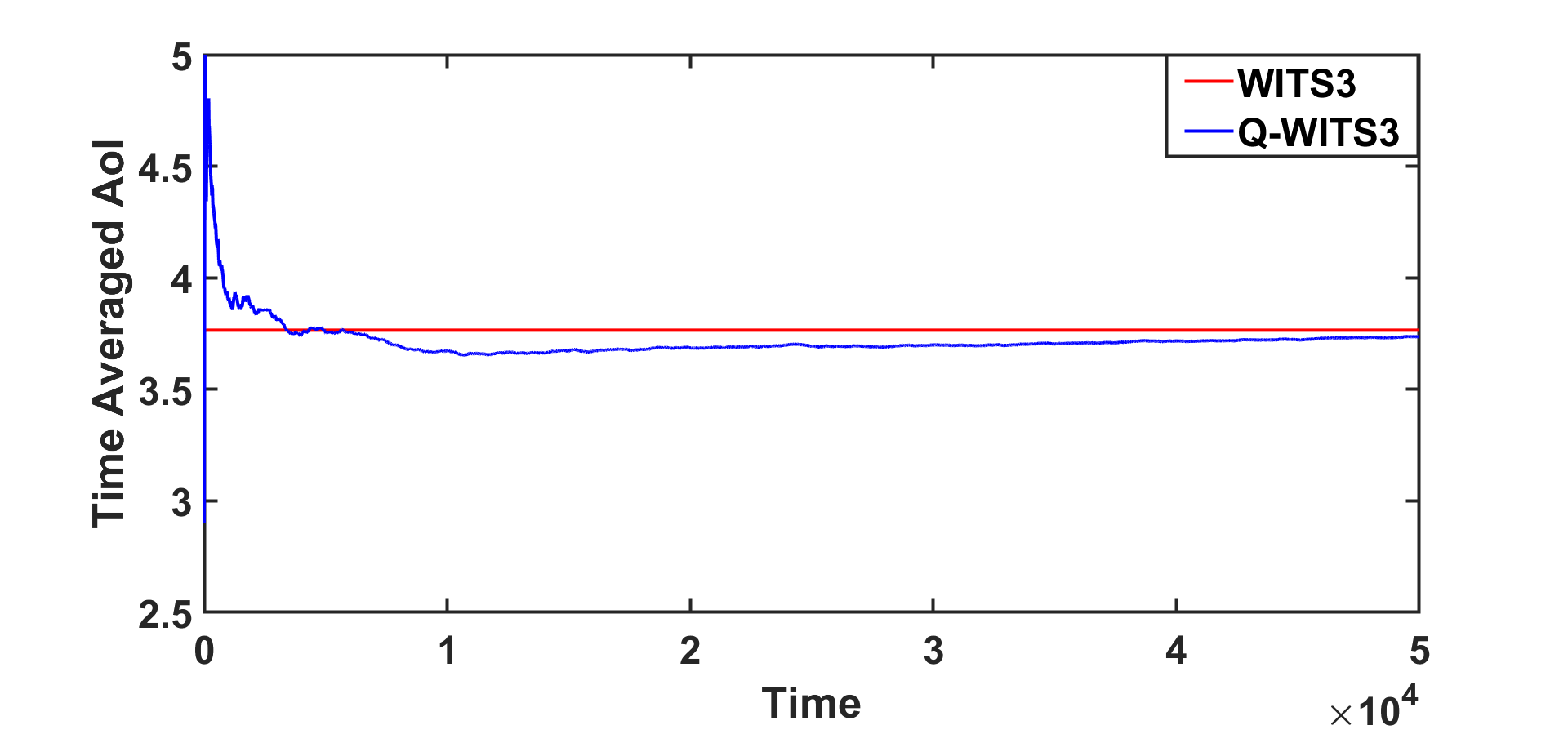}
 \caption{Time performance comparison among proposed Q-WITS3 algorithm, the optimal WITS3 policy, and random policy. }
 \label{multi_avg_age_comp_plot}
 \end{center}
\vspace{-18pt}
\end{figure}

% \vspace{-12pt}
 \section{Conclusion}\label{section:Conclusion}
 %\vspace{-3pt}
We, in this paper, have considered the AoI minimization problem for multiple EH source system where  EH sources send their status updates to a destination node over a common wireless medium. Under a source probing constraint,
we have proposed a Whittle's index and threshold based source scheduling and sampling policy (WITS3 policy)
which uses Whittle’s index for selecting a node for source probing, and further uses a channel quality threshold based rule for sampling and communication. Also, we have proposed a novel Q-WITS3 algorithm that computes the optimal source scheduling and sampling policy for the unknown environment case. Numerical results demonstrate the superiority of WITS3 and also demonstrate that the Q-WITS3 algorithm asymptotically learns the same policy as returned by WITS3. In future, we will extend this work to the multi-sink, multi-hop setting.

% 
% \vspace{-1cm}
% \begin{IEEEbiography}
% [{\includegraphics[width=1in,height=1in,clip,keepaspectratio]{arpan.png}}]
% {Arpan 
% Chattopadhyay} obtained his B.E. in Electronics and Telecommunication  from Jadavpur University, 
% Kolkata, India in 2008, and M.E. and Ph.D in Telecommunication  from Indian Institute of Science, 
% Bangalore, India, in  2010 and 2015, respectively. He is currently working as an assistant professor in the Electrical Engineering department, IIT Delhi. Previously, he held postdoc positions in the Electrical Engineering department, University of Southern California, Los Angeles, USA, and INRIA/ENS Paris, France. 
% His research interests include wireless networks, IoT, cyber-physical systems, networked estimation and control. \end{IEEEbiography}
% 
% \vspace{-1cm}
% \begin{IEEEbiography}
% [{\includegraphics[width=1in,height=1in,clip,keepaspectratio]{ubli.jpg}}]{Urbashi Mitra} is a Professor in the Departments of Electrical
% Engineering and Computer Science, University of Southern California. Previous appointments include Bellcore and the Ohio State
% University. Her honors include: U.S.-U.K. Fulbright, Leverhulme Trust Visiting
% Professorship, Editor-in-Chief IEEE TRANSACTIONS ON MOLECULAR, BIOLOGICAL
% AND MULTISCALE COMMUNICATIONS, IEEE Communications Society
% Distinguished Lecturer, U.S. NAE Galbreth Lectureship, Okawa Foundation
% Award, and an NSF CAREER Award. Her research is in wireless
% communications. \end{IEEEbiography}

\vspace{0mm} 
\bibliographystyle{IEEEtran-mod}
% argument is your BibTeX string definitions and bibliography database(s)
\bibliography{ref_paper.bib}
%\printbibliography
%\vfill

%\appendices

\end{document}